\documentclass[fleqn,a4paper]{amsart}
\usepackage{amsmath,amssymb,amsfonts,amsthm,bbm}
\usepackage{lmodern}
\usepackage[T1]{fontenc}
\usepackage[utf8]{inputenc}

\theoremstyle{plain}
\newtheorem{theorem}{Theorem}[section]
\newtheorem{proposition}[theorem]{Proposition}

\theoremstyle{definition}

\theoremstyle{remark}

\newtheorem{example}{Example}

\newcommand{\pa}{\partial}
\newcommand{\Q}{\mathcal{Q}}
\newcommand{\K}{\operatorname{K}}
\newcommand{\LL}{\operatorname{L}}
\renewcommand{\O}{\mathcal{O}}

\newcommand{\HH}{\mathcal{H}}

\newcommand{\D}{\mathcal{D}}
\renewcommand{\S}{\mathcal{S}}

\newcommand{\A}{\mathcal{A}}
\renewcommand{\L}{\mathcal{L}}

\newcommand{\U}{\mathcal{U}}

\renewcommand{\u}{\mathfrak{u}}

\newcommand{\T}{\operatorname{T}\!}
\renewcommand{\H}{\operatorname{H}\!}
\newcommand{\V}{\operatorname{V}\!}
\newcommand{\Ker}{\operatorname{Ker}}
\newcommand{\Ad}{\operatorname{Ad}}

\newcommand{\Tr}{\operatorname{Tr}}

\newcommand{\diag}{\operatorname{diag}}

\newcommand{\J}{\mathbf{J}}

\newcommand{\0}{{\mathbf 0}}
\newcommand{\II}{\mathcal{I}}
\newcommand{\JJ}{\mathcal{J}}
\newcommand{\ptexp}{\operatorname{exp_{+}\!}}

\newcommand{\ds}{\!\operatorname{d}\!s}
\newcommand{\dd}[1]{\frac{\operatorname{d}}{\operatorname{d}\!#1}\!}

\newcommand{\ket}[1]{|{#1}\rangle}
\newcommand{\bra}[1]{\langle #1 |}
\newcommand{\braket}[2]{\langle #1|#2\rangle}
\newcommand{\ketbra}[2]{|#1\rangle\langle #2|}

\newcommand{\obs}[1]{\hat{#1}}
	
\usepackage{tikz-cd}
\usetikzlibrary{matrix,arrows,decorations.pathmorphing}
	
\begin{document}
\title[A Symmetry Approach to Geometric Phase for Quantum Ensembles]{A Symmetry Approach to Geometric Phase for Quantum Ensembles}
\author{Ole Andersson}
\email{ole.andersson@fysik.su.se}
\author{Hoshang Heydari}
\address{Department of Physics, Stockholm University, 10691 Stockholm, Sweden}
\thanks{The authors thank Ingemar Bengtsson, Erik Sjöqvist, and Marie Ericsson for stimulating discussions. Thanks also go to Johan \"Ohman for suggesting improvements to the text.
HH's research is supported by the Swedish Research Council (VR), grant number 2008-5227.}
\date{\today}

\begin{abstract}
	We use tools from the theory of dynamical systems with symmetries to stratify Uhlmann's standard purification bundle and derive a new connection for mixed quantum states. 
	For unitarily evolving systems, this connection gives rise to the `interferometric' 
	geometric phase of Sj\"{o}qvist \emph{et al.}~[Phys.~Rev.~Lett. 85 2845 - 2849 (2000)],
	and for more generally evolving open systems it gives rise to the generalization of the interferometric geometric phase due to Tong \emph{et al.}~[Phys.~Rev.~Lett. 93 080405 (2004)].
\end{abstract}

\maketitle

\section{Introduction}
Thirty years ago, Berry \cite{Berry1984} demonstrated that a pure quantum state which undergoes an adiabatic evolution picks up a phase of purely geometric origin. Almost immediately, Simon \cite{Simon1983} put Berry's discovery into a mathematical context, identifying Berry's phase as the argument of a holonomy in a certain line bundle over the parameter manifold of the driving Hamiltonian. Shortly thereafter, Aharonov and Anandan \cite{Aharonov_etal1987} removed the adiabaticity requirement by defining the geometric phase as the argument of a holonomy in the generalized Hopf bundle.
Since then, an enormous number of papers on the subject have been published, and geometric phase has been recognized as a concept of central importance in quantum physics \cite{Shapere_etal1989,Markovski_etal1989,Zwanziger_etal1990, Bohm_etal2003, Chruscinski_etal2004, Sjoqvist2008}.

In general, an actual quantum system does not allow its state to be modeled by a single vector.
Because of entangling interaction with an environment, or imperfection in its preparation, 
the state has a mixed character and can at best be modeled by a density operator. 
Uhlmann was the first to consider the notion of geometric phase for mixed quantum states, using his elegant and generally applicable ``purification theory'' \cite{Uhlmann1986,Uhlmann1989,Uhlmann1991}. More recently, Sj\"oqvist~\emph{et al.}~\cite{Sjoqvist_etal2000} have derived a geometric phase for mixed quantum states in an experimental context. This geometric phase, which differs from that of Uhlmann \cite{Slater2001,Slater2002,Mousolou_etal2011}, was initially only defined for quantum systems which evolve unitarily. But Tong \emph{et al.}~\cite{Tong_etal2004} soon extended its definition to apply to quantum systems undergoing a more general, in fact arbitrary, evolution.

Although the relationship between the phases of Uhlmann and Sj\"oqvist-Tong \emph{et al.}~has been studied extensively, e.g. see \cite{Tidstrom_etal2003,Ericsson_etal2003,Rezakhani_etal2006,Shi_etal2006,Aberg_etal2007,Mousolou_etal2011}, a mathematical construction like that of Uhlmann, leading to the geometric phase of Sj\"oqvist-Tong \emph{et al.}, has until now been missing \cite{Chruscinski_etal2004, Shi_etal2006}. In this paper we conduct a detailed geometric analysis of Uhlmann's standard purification bundle using tools from the theory of dynamical systems with symmetries, and we derive a new connection for this bundle which gives rise to Sj\"oqvist-Tong \emph{et al.}'s geometric phase.

The outline of the paper is as follows. In Section \ref{operator spaces} we introduce several operator spaces with unitary group actions, as well as Uhlmann's standard purification bundle. In Section \ref{section:goemtricphase} we briefly discuss Uhlmann's, Sj\"oqvist \emph{et al.}'s, and Tong \emph{et al.}'s definitions of geometric phase. In Section \ref{section:dynamics} we apply the theory of dynamical systems with symmetries to Uhlmann's theory of purification and show how the standard purification bundle can be reduced to a principal fiber bundle with a connection that reproduces the geometric phase of Sj\"oqvist \emph{et al.} In Section \ref{section:open} we extend the construction initiated in Section \ref{section:dynamics} and derive a connection for the standard purification bundle which reproduces the geometric phase of Tong \emph{et al.} 

\section{Some operator spaces with unitary group actions}\label{operator spaces}
Let $\HH$ be a finite dimensional complex Hilbert space and $\L^\circ(\HH)$ be the space of invertible operators on $\HH$
equipped with the Hilbert-Schmidt Hermitian product. Furthermore, let $\S^\circ(\HH)$ be the space of invertible 
operators on $\HH$ with unit norm and $\D^\circ(\HH)$ be the space of invertible density operators on $\HH$.
The group of unitary operators on $\HH$, which we denote by $\U(\HH)$, acts on $\L^\circ(\HH)$ from the left and from the right by operator composition:
\begin{alignat}{2}
	&L_U : \L^\circ(\HH)\to\L^\circ(\HH), &\quad &L_U(\psi) = U\psi, \label{left} \\ 
	&R_U : \L^\circ(\HH)\to\L^\circ(\HH), &\quad &R_U(\psi) = \psi U, \label{right}
\end{alignat}
and on $\D^\circ(\HH)$ by left and right operator conjugation:
\begin{alignat}{2}
	&l_U : \D^\circ(\HH)\to\D^\circ(\HH), &\quad &l_U(\rho) = U\rho U^\dagger, \label{leftconj}\\
	&r_U : \D^\circ(\HH)\to\D^\circ(\HH), &\quad &r_U(\rho) = U^\dagger\rho U. \label{rightconj}
\end{alignat}
Moreover, the actions \eqref{left} and \eqref{right} preserve $\S^\circ(\HH)$, and for each $\psi$ 
in $\S^\circ(\HH)$ the composition $\psi\psi^\dagger$ is a density operator on $\HH$, say 
$\psi\psi^\dagger =\rho$, with $(U\psi)(U\psi)^\dagger=U\rho U^\dagger$ and $(\psi U)(\psi U)^\dagger=\rho$. 
Thus, the map 
\begin{equation}
	\Pi:\S^\circ(\HH)\to \D^\circ(\HH),\quad \Pi(\psi)=\psi\psi^\dagger \label{standard}
\end{equation}
is such that the following two diagrams are commutative:
\begin{center}
	\begin{tikzcd}
		\S^\circ(\HH) \arrow{d}[swap]{\Pi} \arrow{r}{L_U} & \S^\circ(\HH)\arrow{d}{\Pi} \\
		\D^\circ(\HH) \arrow{r}{l_U} & \D^\circ(\HH)
	\end{tikzcd}
	\qquad
	\begin{tikzcd}[column sep=small]
		\S^\circ(\HH)  \arrow{rd}[swap]{\Pi} \arrow{rr}{R_U} & {} & \S^\circ(\HH)\arrow{dl}{\Pi} \\
		{} & \D^\circ(\HH) & {}
	\end{tikzcd}
\end{center}
A key fact used in the present paper is that $\Pi$ is a principal fiber bundle with symmetry group $\U(\HH)$ whose right action is given by the restriction of \eqref{right} to $\S^\circ(\HH)$.
We recommend \cite{Kobayashi_etal1996III} as a general reference on the theory of principal fiber bundles.
	
We also need notation for four additional unitary group actions, or, rather, representations.
Thus let $\u(\HH)$ be the Lie algebra of $\U(\HH)$, that is, the Lie algebra consisting of all skew-Hermitian operators on $\HH$,
and $\u(\HH)^*$ be the space of all real-valued linear functions on $\u(\HH)$.
The left and right adjoint representations of $\U(\HH)$ on $\u(\HH)$ are
\begin{alignat}{2}
	U&\mapsto \Ad_U,&\quad &\Ad_U\xi = U\xi U^\dagger,\\
	U&\mapsto \Ad_{U^\dagger},&\quad &\Ad_{U^\dagger} \xi = U^\dagger \xi U,
\end{alignat}
respectively, and the left and right coadjoint representations of $\U(\HH)$ on $\u(\HH)^*$ are
\begin{alignat}{2}
	U&\mapsto \Ad^*_{U^\dagger}, &\quad &\Ad^*_{U^\dagger}(\mu) = \mu\circ\Ad_{U^\dagger},\\
	U&\mapsto \Ad^*_U,&\quad &\Ad^*_{U}(\mu) = \mu\circ \Ad_{U}.
\end{alignat}
	
Finally, evolving operators will be represented by parameterized curves which are
assumed to be smooth and defined on an unspecified interval $0\leq t\leq \tau$.
	
\section{Geometric phase for mixed quantum states}\label{section:goemtricphase}
Here we briefly describe the two main approaches to geometric phase for mixed states mentioned in the introduction. 
For issues concerning the experimental verifiability of these phases see \cite{Du_etal2003, Ericsson_etal2005, Klepp_etal2005, Ghosh_etal2006}. 

\subsection{The standard purification bundle and Uhlmann's geometric phase}
In this paper we adopt Uhlmann's terminology and call the operators in $\S^\circ(\HH)$ \emph{purifications}, and the principal fiber bundle $\Pi$
\emph{the standard purification bundle}. Thus, a purification of a density operator $\rho$ is a $\psi$ in $\S^\circ(\HH)$ such that $\Pi(\psi)=\rho$, and two purifications $\psi$ and $\phi$ purify the same density operator if and only if $\phi=\psi U$ for some unitary operator $U$ on $\HH$.
	
The vertical bundle of $\Pi$ is the vector bundle $\V\S^\circ(\HH)$ whose fibers are the kernels of the differential of $\Pi$,
and the horizontal bundle of $\Pi$ is the vector bundle $\H\S^\circ(\HH)$ whose fibers are the orthogonal complements of the fibers of $\V\S^\circ(\HH)$ with respect to the real part of the Hilbert-Schmidt product. 
The vertical and horizontal bundles split the tangent bundle of $\S^\circ(\HH)$ into a direct sum, $\T\S^\circ(\HH)=\V\S^\circ(\HH)\oplus\H\S^\circ(\HH)$; explicitly, the fibers at $\psi$ of these bundles are
\begin{align}
	\V_\psi\S^\circ(\HH)&=\{X\in \T_\psi\S^\circ(\HH):  X\psi^\dagger+\psi X^\dagger =0\},\label{vertical}\\
	\H_\psi\S^\circ(\HH)&=\{X\in \T_\psi\S^\circ(\HH):  X^\dagger\psi-\psi^\dagger X =0\}.\label{horizontal}
\end{align}
	
The horizontal bundle is a connection on $\Pi$.
Therefore, all curves in $\D^\circ(\HH)$ can be lifted to horizontal curves in $\S^\circ(\HH)$.
This means that for each curve $\rho(t)$ in $\D^\circ(\HH)$, there is a curve $\psi(t)$ in $\S^\circ(\HH)$ which projects onto $\rho(t)$ and whose velocity field is a curve in $\H\S^\circ(\HH)$.
Furthermore, this curve is unique if its initial point in the fiber over $\rho(0)$ is specified, which can be chosen arbitrarily.
The Uhlmann geometric phase of $\rho(t)$ is then defined as
\begin{equation}
	\Gamma_g[\rho(t)]=\arg\Tr\left(\psi(0)^\dagger\psi(\tau)\right).
\end{equation}
	
\subsection{Parallel transporting evolution operators and the geometric phase of Sj\"{o}qvist et al.}
Sj\"{o}qvist \emph{et al.}~\cite{Sjoqvist_etal2000} have proposed a geometric phase for unitarily evolving density operators which
is different from Uhlmann's geometric phase.
Consider a density operator $\rho$ with non-degenerate spectrum.
Following Sj\"{o}qvist \emph{et al.}~we say that 
a unitary evolution operator \emph{parallel transports} $\rho$ if the trajectories of the eigenstates of $\rho$ are horizontal curves in the sense of Aharonov and Anandan \cite{Aharonov_etal1987}.
In other words, $U(t)$ parallel transports $\rho$ if for every eigenvector $\ket{\psi_k}$ of $\rho$, the curve $\ket{\psi_k(t)}=U(t)\ket{\psi_k}$ satisfies $\braket{\psi_k(t)}{\dot{\psi_k}(t)}=0$.
For such a parallel transported density operator we define the geometric phase to be
\begin{equation}
	\gamma_{g}[\rho(t)]=\arg\Tr\left(U(\tau)\rho\right).\label{geophase}
\end{equation} 
	
If the density operator $\rho$ has a degenerate spectrum, the parallel transport condition needs to be slightly modified  \cite{Tong_etal2005}.
In this case we say that $U(t)$ is parallel transporting provided that $\braket{\psi_k(t)}{\dot{\psi_l}(t)}=0$ for any pair of eigenvectors
$\ket{\psi_k}$ and $\ket{\psi_l}$ of $\rho$ with the same eigenvalue.
The geometric phase is still defined by the formula \eqref{geophase}.
	
It is known that for unitarily evolving mixed states, $\Gamma_{g}$ and $\gamma_{g}$ don't necessarily coincide, see \cite{Slater2001,Slater2002}.
This is what one would expect since unitary evolution operators do not generate horizontal curves:
\begin{proposition}\label{skillnad}
A non-stationary unitary evolution operator will not generate a horizontal curve in the total space of the standard purification bundle.
\end{proposition}
\begin{proof}
Assume that $\psi(t)=U(t)\psi$ is a horizontal curve in $\S^\circ(\HH)$.
Let $\xi(t)$ be the curve of skew-Hermitian operators defined by $\dot{U}(t)=U(t)\xi(t)$.
Then $\dot{\psi}(t)=U(t)\xi(t)\psi$, and according to \eqref{horizontal},
\begin{equation}
	\begin{split}
		0
		&=\dot{\psi}(t)^\dagger\psi(t)-\psi(t)^\dagger \dot{\psi}(t)\\
		&=\psi^\dagger \xi(t)^\dagger U(t)^\dagger U(t) \psi-\psi^\dagger U(t)^\dagger U(t) \xi(t)\psi\\
		&=-2\psi^\dagger \xi(t)\psi.
	\end{split}
\end{equation}
But then $\xi(t)$ vanishes identically because $\psi$ is invertible.
\end{proof}
\noindent The following example by Slater illustrates the fact that the geometric phases of Uhlmann and Sj\"oqvist \emph{et al.}~in general are different.
	
\begin{example}[Slater \cite{Slater2002}]\label{difference}
Slater considered a mixed qubit state which in Bloch ball coordinates is represented by
\begin{align}
	\rho=\frac{1}{2}(1+r\mathbf{r}\cdot\boldsymbol{\sigma}),\quad r\mathbf{r}=r(\sin\theta\cos\phi,\sin\theta\sin\phi,\cos\theta)^T,
\end{align}
and showed that if $\rho$ is initially such that $\mathbf{r}=(0,0,1)^T$, and is unitarily driven in such a way that the trajectory of the the Bloch vector $\mathbf{r}$ is a geodesic triangle on the Bloch sphere with vertices at $(0,0)$, $(\theta_1,\phi_1)$, and $(\theta_2,\phi_2)$, then the ratio of the tangents of the geometric phases of Sj\"oqvist \emph{et al.}~and Uhlmann is
\begin{equation}
	\frac{\tan\gamma_g[\rho(t)]}{\tan\Gamma_g[\rho(t)]}=1+\frac{4-10r^2+6r^4}{r^2((1+\cos\theta_1)(1+\cos\theta_2)+\cos(\phi_1-\phi_2)\sin\theta_1\sin\theta_2)}.
\end{equation}
Slater also showed that if, instead, $\rho$ is driven in such a way that $\mathbf{r}$ is rotated once about $\mathbf{n}=(0,\sin\xi,\cos\xi)$,
then 
\begin{equation}
	\frac{\tan\gamma_g[\rho(t)]}{\tan\Gamma_g[\rho(t)]}=\frac{\pi\chi \tan(\pi\cos\xi)\coth(\pi\chi)}{\pi\cos\xi},
\end{equation}
where
\begin{equation}
	\chi=\sqrt{\frac{r^2-2-r^2\cos(2\xi)}{2}}.
\end{equation}
In each case the parameters can be chosen to make the ratio different from $1$.
\end{example}
	
\subsection{Geometric phase for open quantum systems}
Uhlmann's geometric phase is defined for arbitrarily evolving quantum systems
while, in its original form, the definition of geometric phase by Sj\"oqvist \emph{et al.}~only applies to unitarily evolving systems. This limitation was overcome by Tong \emph{et al.}~\cite{Tong_etal2005} who extended the definition \eqref{geophase} to include quantum systems with arbitrary evolution. Here we describe this extension.
	
Consider a curve $\rho(t)$ in $\D^\circ(\HH)$ modeling the evolution of a state represented by a density operator $\rho=\rho(0)$, and let
\begin{equation}
    \rho(t)=\sum_k p_k(t)\ketbra{\psi_k(t)}{\psi_k(t)}
\end{equation}
be a spectral decomposition of $\rho(t)$ which is such that the normalized eigenstates $\ket{\psi_k(t)}$ vary smoothly with time.
We shall say that $\rho$ is parallel transported (by whatever the evolution operator of the system may be) if, for every $t$, the eigenstates satisfy the parallelism condition 
\begin{equation}\label{pc}
	p_k(t)=p_l(t)\implies \braket{\psi_k(t)}{\dot{\psi_l}(t)}=0.
\end{equation}
In this case, we define the geometric phase of $\rho$ to be
\begin{equation}\label{opengp}
	\gamma_g[\rho(t)]=\arg\sum_k\sqrt{p_k(0)p_k(\tau)}\braket{\psi_k(0)}{\psi_k(\tau)}.
\end{equation}
For systems where the eigenvalues $p_k$ do not vary over time, \eqref{opengp} reduces to \eqref{geophase}.
	
\section{Hamiltonian dynamics in the standard purification bundle and geometric phase for unitarily evolving quantum states}\label{section:dynamics}
In this section we show how the standard purification bundle can be reduced to a bundle in which parallel transporting unitary evolution operators generate horizontal curves of purifications. We also reproduce the geometric phase of Sj\"oqvist \emph{et al.}
	
\subsection{Right unitarily symmetric Hamiltonian dynamics in the standard purification bundle}
The imaginary part of the Hilbert-Schmidt product, multiplied by $2\hbar$, is a symplectic form on $\L^\circ(\HH)$,
\begin{equation}
	\Omega(X,Y)=-i\hbar\Tr(X^\dagger Y-Y^\dagger X),
\end{equation}
and the action \eqref{right} is symplectic,
\begin{equation}
	R_U^*\Omega =\Omega.
\end{equation}
A \emph{right unitarily symmetric Hamiltonian dynamical system} on $\L^\circ(\HH)$ is an evolution equation of the form $\dot{\psi}=X_{H}(\psi)$, where $X_H$ is the Hamiltonian vector field of a real-valued function $H$ on $\L^\circ(\HH)$ which is constant along the orbits of the action \eqref{right}.
For such systems, if $\psi(t)$ is the solution that emanates from $\psi$, $\psi(t)U$ is the solution that emanates from $\psi U$.

\begin{example}[The Schr\"odinger equation]\label{shroedinger}
Suppose that $\obs{H}$ is a Hermitian operator on $\HH$.
Then $\obs{H}$ defines a Hermitian operator on $\L^\circ(\HH)$, which sends $\psi$ to the composition $\obs{H}\psi$. 
Let $H$ be the associated expectation value function, $H(\psi)=\Tr(\psi^\dagger\obs{H}\psi)$. Then 
\begin{equation}
	X_H(\psi)=\frac{1}{i\hbar}\obs{H}\psi.
\end{equation}
Thus the dynamical system of $X_H$ is Schr\"odinger's equation. Clearly, $H$ is right unitarily invariant.
\end{example} 

Next we define a momentum map for the right action \eqref{right}. 
Thus for each $\xi$ in $\u(\HH)$ define a smooth function $J_\xi$ on $\L^\circ(\HH)$ by 
$J_\xi(\psi)=i\hbar\Tr(\psi^\dagger\psi\xi)$.
The Hamiltonian vector field of $J_\xi$ is the fundamental vector field of the right $\U(\HH)$-action corresponding to $\xi$.
That is,
\begin{equation}
	dJ_\xi(X)=\Omega(\hat{\xi},X),
\end{equation}
where
\begin{equation}
	\hat{\xi}(\psi)=\dd{t}\Big[R_{\exp(t\xi)}(\psi)\Big]_{t=0}=\psi\xi.
\end{equation}
Moreover, the assignment $\xi\mapsto J_\xi$ is unitarily equivariant, $J_{\Ad_{U}\xi}=J_\xi\circ R_U$. Thus
\begin{equation}
	\J:\L^\circ(\HH) \to \u(\HH)^*,\quad \J(\psi)\xi=J_\xi(\psi)=i\hbar\Tr(\psi^\dagger\psi\xi)
\end{equation}
is a coadjoint-equivariant momentum map for the right action of $\U(\HH)$ on $\L^\circ(\HH)$.
For right unitarily symmetric Hamiltonian dynamical systems we have the following Noetherian theorem:
\begin{theorem}\label{noether}
The flow of each right unitarily symmetric Hamiltonian dynamical system on $\L^\circ(\HH)$ preserves the level sets of $\J$.
\end{theorem}
\begin{proof}
Suppose that $H$ is a real-valued function on $\L^\circ(\HH)$ which is constant along the orbits of the right action by $\U(\HH)$.
Then $d\J(X_H)\xi=-dH(\hat{\xi})=0$ for every $\xi$ in $\u(\HH)$. Thus $X_H$ is everywhere tangential to the level sets of $\J$.
\end{proof}
\noindent Next we verify that the level sets of $\J$ are smooth manifolds:
\begin{proposition}\label{smooth}
Every $\psi$ in $\L^\circ(\HH)$ is a regular point of $\J$.
\end{proposition}
\begin{proof}
Recall that $\psi$ in $\L^\circ(\HH)$ is a regular point of $\J$ if the differential of $\J$ maps $\T_\psi\L^\circ(\HH)$ onto $\u(\HH)^*$. Now, assume that this is not the case. 
According to Riesz' lemma \cite[Th II.4, p 43]{Reed_etal1980}, there exists a non-zero $\xi$ in $\u(\HH)$ such that $d\J(X)\xi=0$ for all $X$ in $\T_\psi\L^\circ(\HH)$. However, since $d\J(X)\xi=\Omega(\hat{\xi}(\psi),X)$ and $\Omega$ is non-degenerate, this implies the contradictory conclusion that $\xi=0$.
\end{proof}
	
It follows from Theorem \ref{noether} and Example \ref{shroedinger}, or the very definition of $\J$ for that matter, 
that if $U(t)$ is a unitary evolution operator and $\psi$ is an invertible operator on $\HH$, the curve 
$\psi(t)=U(t)\psi$ will be completely contained within the level set $\J^{-1}(\J(\psi))$, see Figure \ref{figure:levelset}.
In the next paragraph we will show that if $\psi$ purifies $\rho$, then this level set is the total 
space of a principal fiber bundle over the space of density operators on $\HH$ which are 
unitarily equivalent to $\rho$, that is, the orbit of $\rho$ under the action \eqref{leftconj}.
\begin{figure}
\centering
\includegraphics[width=0.55\textwidth]{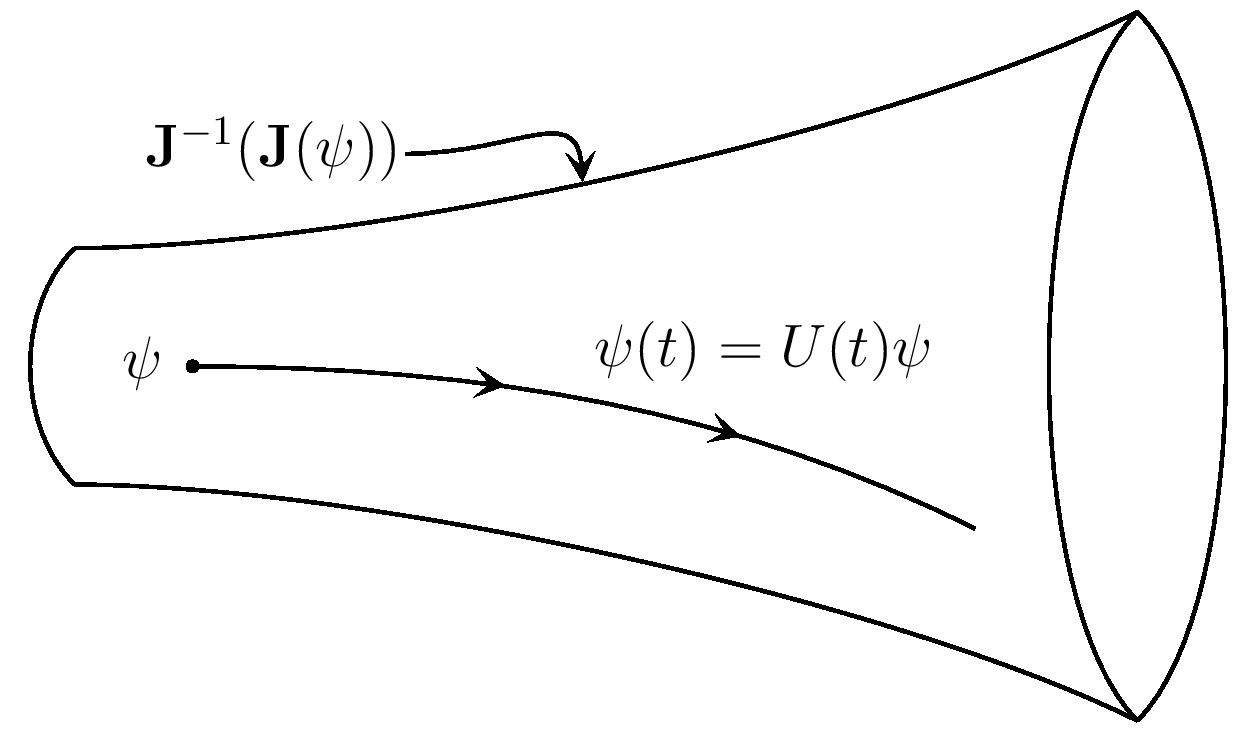}
\caption{The trajectory of a unitarily evolving invertible operator on $\HH$ is contained in a single level set of $\J$.}
\label{figure:levelset}
\end{figure}
We end this paragraph by showing that the left action \eqref{left} is transitive on the level sets of $\J$.
\begin{proposition}
Let $\psi$ and $\phi$ be invertible operators on $\HH$.
If $\J(\psi)=\J(\phi)$, then $\phi=U\psi$ for some $U$ in $\U(\HH)$.
\end{proposition}
\begin{proof}
The condition $\J(\psi)=\J(\phi)$ is equivalent to $\psi^\dagger\psi=\phi^\dagger\phi$, which in turn implies that $U=(\phi^\dagger)^{-1}\psi^\dagger$ is unitary and $\phi=U\psi$.
\end{proof} 
	
\subsection{Reduced standard purification bundles}
A unitarily evolving invertible density operator remains in a single orbit of the left conjugation 
action of $\U(\HH)$ on $\D^\circ(\HH)$. We write $\D(\rho)$ for the orbit that contains $\rho$. 
Moreover, we call two density operators that belong to the same orbit \emph{isospectral} since they have the same eigenvalues.
	
For each purification $\psi$ let $\S(\psi)$ be the level set $\J^{-1}(\J(\psi))$. 
Because all the members of $\S(\psi)$ have unit norm it follows from Proposition \ref{smooth} 
that $\S(\psi)$ is a immersed submanifold of $\S^\circ(\HH)$. Moreover, by Proposition \ref{noether}, 
the flow of every right unitarily symmetric Hamiltonian dynamical system on $\L^\circ(\HH)$ preserves 
$\S(\psi)$. The main result of this paragraph says that if  $\psi$ purifies $\rho$, the restriction 
of the standard purification $\Pi$ to $\S(\psi)$ is a principal fiber bundle over $\D(\rho)$. 
\begin{proposition}\label{onto}
If $\psi$ purifies $\rho$, then $\Pi$ maps $\S(\psi)$ onto $\D(\rho)$.
\end{proposition} 
	
\begin{proof}
Let $\phi$ be a purification in $\S(\psi)$.
Then $\Pi(\phi)$ and $\rho$ are isospectral as $\J(\phi) = \J(\psi)$ is equivalent to $\phi^\dagger\phi = \psi^\dagger\psi$, and $\psi^\dagger\psi$ and $\psi\psi^\dagger$ are isospectral. Thus, $\Pi(\phi)$ belongs to $\D(\rho)$.
Next, let $\sigma$ be an arbitrary density operator in $\D(\rho)$. Then $\sigma = U\rho U^\dagger$ for some unitary operator $U$ on $\HH$.
Now, $U\psi$ belongs to $\S(\psi)$, and $\Pi(U\psi)=\sigma$.
We conclude that $\Pi$ maps $\S(\psi)$ onto $\D(\rho)$.
\end{proof}
For each purification $\psi$ let $\U(\psi)$ be the right isotropy group of $\J(\psi)$,
\begin{equation}
	\U(\psi)
	=\{U\in\U(\HH):\Ad_U^*\J(\psi)=\J(\psi)\}
	=\{U\in\U(\HH): U\psi^\dagger\psi=\psi^\dagger\psi U\}.
\end{equation}
This group acts freely on $\S(\psi)$ from the right,
and properly discontinuously because it is compact.
Thus the coset space $\S(\psi)/\U(\psi)$ is a manifold, and the quotient map $\phi\mapsto [\phi]$ is a principal fiber bundle with right acting symmetry group $\U(\psi)$. 
\begin{proposition}\label{diffeomorphism}
If $\psi$ purifies $\rho$, then $[\phi]\mapsto \phi\phi^\dagger$ is a diffeomorphism
from  $\S(\psi)/\U(\psi)$ onto $\D(\rho)$.
\end{proposition}
\begin{proof}
Clearly, the assignment $[\phi]\mapsto \phi\phi^\dagger$ is a well-defined and smooth map
from  $\S(\psi)/\U(\psi)$ onto $\D(\rho)$. To see that it is also injective assume that $\phi$ and $\chi$ in $\S(\psi)$
are such that $\chi\chi^\dagger=\phi\phi^\dagger$. Then $\chi=\phi U$ for a unitary operator on $\HH$. Now,
\begin{equation}
    U\psi^\dagger\psi=U\chi^\dagger\chi=UU^\dagger\phi^\dagger\phi U=\phi^\dagger\phi U=\psi^\dagger\psi U.
\end{equation}     
Thus $U$ belongs to $\U(\psi)$.
\end{proof}
	
According to Propositions \ref{onto} and \ref{diffeomorphism} we have a commutative diagram 
\begin{center}
	\begin{tikzcd}
	{} & \S(\psi) \arrow{dl}[swap]{\phi\mapsto [\phi]} \arrow{dr}{\Pi} & {} \\
	\S(\psi)/\U(\psi) \arrow{rr}{[\phi]\mapsto \phi\phi^\dagger} & {} & \D(\rho)
	\end{tikzcd}
\end{center}
in which the horizontal map is a diffeomorphism. Therefore, the restriction of $\Pi$ to $\S(\psi)$ is a principal fiber bundle with right acting symmetry group $\U(\psi)$.
Moreover, the isomorphism class of this bundle does not depend on the choice of $\psi$:
\begin{proposition}\label{relation}
If $\psi$ and $\phi$ purifies $\rho$, and $\phi=\psi U$, then $\U(\phi)=U^\dagger\U(\psi) U$, and $R_U$ restricts to an equivariant bundle map from $\S(\psi)$ onto $\S(\phi)$. In other words, the following diagram is commutative:
\begin{center}
	\begin{tikzcd}\label{diagrammet}
	\S(\psi) \arrow{dr}[swap]{\Pi} \arrow{rr}{R_U} & {} & \S(\phi)\arrow{dl}{\Pi} \\
	{} & \D(\rho) & {}
	\end{tikzcd}
\end{center}
\end{proposition}
\begin{proof}
The proof is straightforward, so we confine ourselves to showing that $R_U$ is equivariant.
Thus let $V$ be a unitary in $\U(\psi)$. Then,
\begin{equation}
	R_{U^\dagger V U}(R_U(\chi))=\chi UU^\dagger VU=\chi VU=R_U(R_V(\chi)).
\end{equation}
\end{proof}
\noindent Figure \ref{figure:redbun} illustrates the relation between the restrictions of $\Pi$ to $\S(\psi)$ and to $\S(\phi)$.
\begin{figure}
\centering
\includegraphics[width=0.60\textwidth]{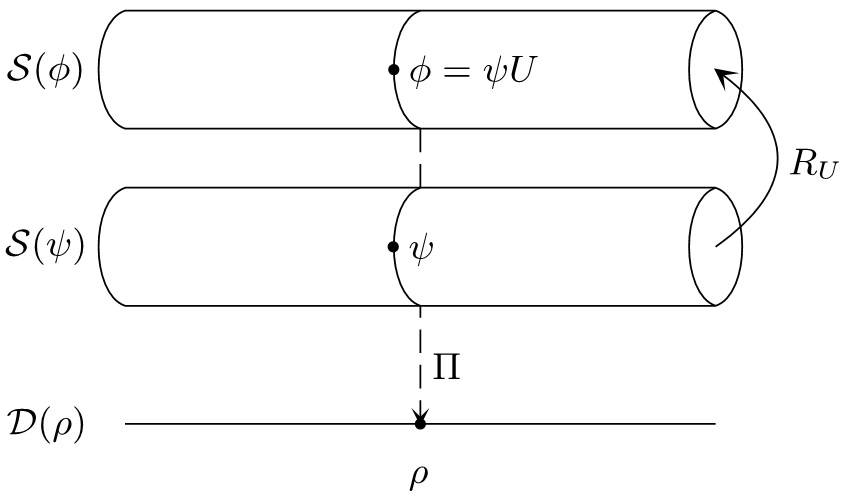}
\caption{Different purifications of the same density operator give rise to isomorphic reduced standard purification bundles.}
\label{figure:redbun}
\end{figure}
	
Next, we will equip the reduced standard purification bundles with connections which are such that if $U(t)$ parallel transports $\rho$ and $\psi(t)=U(t)\psi$, where $\psi$ purifies $\rho$, then $\psi(t)$ is horizontal. The geometric phase defined at \eqref{geometriphase} below thus generalizes the definition by Sj\"oqvist \emph{et al.}
Throughout the rest of this section we assume that $\rho$ is a given invertible density operator on $\HH$ and that $\psi$ is a purification of $\rho$.
	
\subsection{Geometric phase}
The real part of the Hilbert-Schmidt inner product, multiplied by $2\hbar$, is a Riemannian metric on $\L^\circ(\HH)$,
\begin{equation}\label{metric}
	G(X,Y)=\hbar\Tr(X^\dagger Y+ Y^\dagger X).
\end{equation}
This metric restricts to a metric on $\S(\psi)$ which is invariant under the action of the symmetry group $\U(\psi)$. We define the vertical and horizontal bundles over $\S(\psi)$ to be the 
subbundles $\V\S(\psi)=\Ker d(\Pi|_{\S(\psi)})$ and $\H\S(\psi)=\V\S(\psi)^\bot$
of the tangent bundle $\T\S(\psi)$. Here $d(\Pi|_{\S(\psi)})$ is the differential of the restriction of $\Pi$ to $\S(\psi)$, and $^\bot$ denotes the fiberwise orthogonal complement in $\T\S(\psi)$ with respect to $G$. The following proposition, illustrated in Figure \ref{levelsets}, is an infinitesimal version of Proposition \ref{skillnad}:
\begin{proposition}\label{infskillnad}
$\T_\psi\S(\psi)$ and Uhlmann's horizontal space $\H_\psi\S^\circ(\HH)$ have no common non-zero vector.
\end{proposition}
\begin{proof}
Every vector $X$ in $\T_\psi\S(\psi)$ satisfies $X^\dagger\psi+\psi^\dagger X=0$,
and every vector $X$ in $\H_\psi\S^\circ(\HH)$ satisfies $X^\dagger\psi-\psi^\dagger X=0$.
As $\psi$ is invertible, $X=0$ is the only common solution to these equations.
\end{proof}
\begin{figure}
	\centering
	\includegraphics[width=0.60\textwidth]{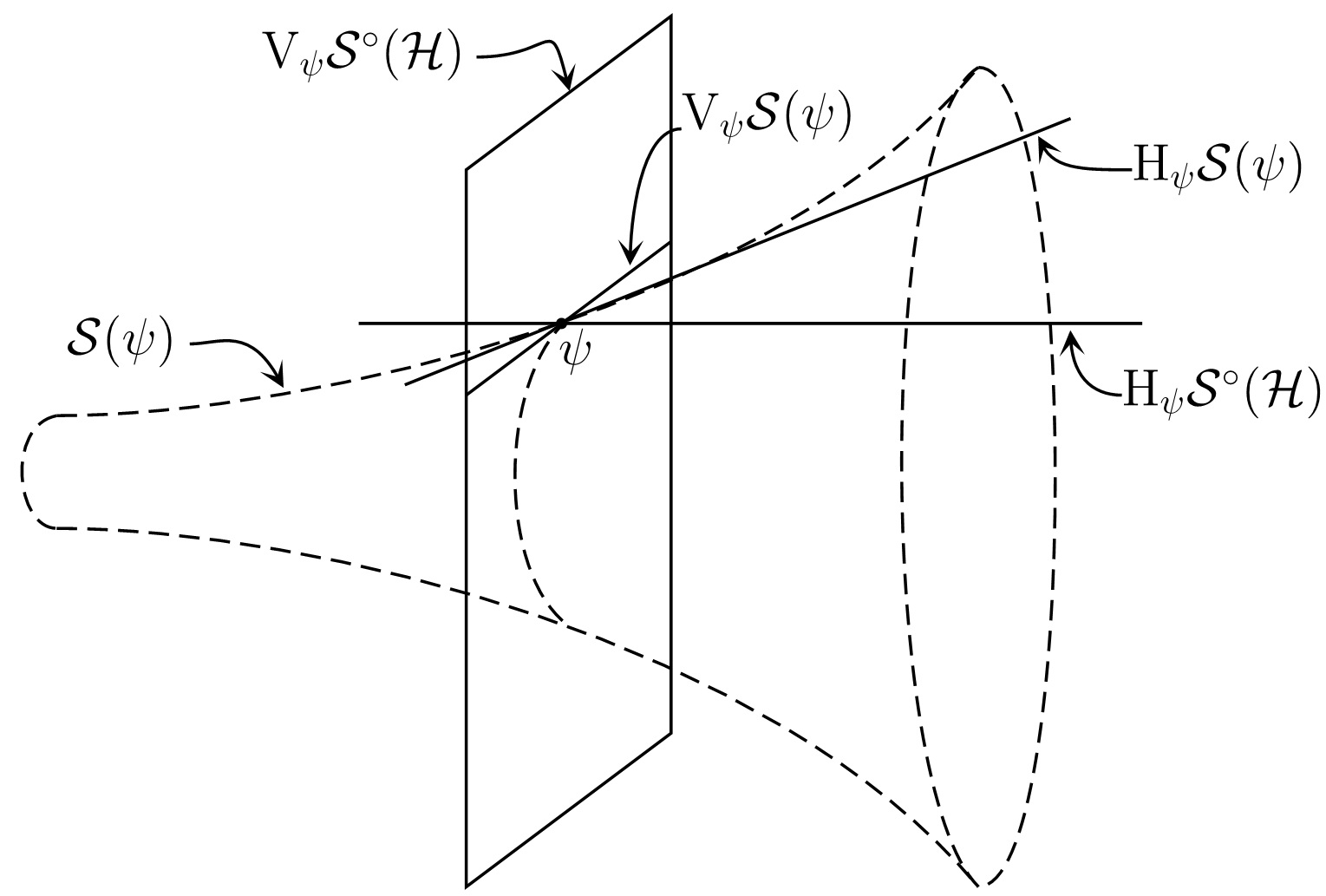}
	\caption{The tangent bundle of $\S(\psi)$ and Uhlmann's horizontal bundle have no common non-zero vectors.}
	\label{levelsets}
\end{figure}
Let $\rho(t)$ be a curve of density operators in $\D(\rho)$.
We define the geometric phase of $\rho(t)$ to be
\begin{equation}\label{geometriphase}
	\gamma_g[\rho(t)]=\arg\Tr\left(\psi(0)^\dagger\psi(\tau)\right),
\end{equation}
where $\psi(t)$ is any horizontal lift of $\rho(t)$ to $\S(\psi)$.
Clearly, the geometric phase is invariant under the action of $\U(\psi)$. Moreover, it does not depend on 
the choice of level set $\S(\psi)$ because every bundle map $R_U$ maps $\S(\psi)$ isometrically 
onto $\S(\psi U)$ with respect to $G|_{\S(\psi)}$ and $G|_{\S(\psi U)}$, and therefore preserves 
horizontality of curves. (See also Proposition \ref{preservation} below.) 
We will in the next paragraph show that for parallel transported density 
operators, \eqref{geometriphase} agrees with the definition of geometric phase 
\eqref{geophase} by Sj\"oqvist \emph{et al.}
	
\subsection{Construction of horizontal dynamics}
Let $\u(\psi)$ be the Lie algebra of the symmetry group $\U(\psi)$, that is, the Lie algebra of all skew-Hermitian operators on $\HH$ which commutes with $\psi^\dagger\psi$. The fundamental vector fields of the symmetry group action on $\S(\psi)$ yield canonical isomorphisms between $\u(\psi)$ and the fibers in $\V\S(\psi)$,
\begin{equation}
	\u(\psi)\ni\xi\mapsto\hat{\xi}(\phi)\in\V_\phi\S(\psi).
\end{equation}
Furthermore, $\H\S(\psi)$ is the kernel bundle of the mechanical connection form 
\begin{equation}
	\A:\T\S(\psi)\to\u(\psi),\quad \A_{\phi}=\II_{\phi}^{-1}\JJ_{\phi},
\end{equation} 
where 
\begin{alignat}{2}
	\II_\phi&:\u(\psi)\to \u(\psi)^*, &\quad  &\II_{\phi}(\xi)\eta=G(\hat\xi(\phi),\hat{\eta}(\phi)),\\
	\JJ_{\phi}&:\T_\phi\S(\psi)\to \u(\psi)^*,&\quad &\JJ_{\phi}(X)\eta=G(X,\hat\eta(\phi)).
\end{alignat}
If,  with respect to any ordering of the eigenvalues of $\psi^\dagger\psi$, we write 
$P_j$ for the orthogonal projection onto the $j^\text{th}$ eigenspace of $\psi^\dagger\psi$, the 
connection adopts the following form:
\begin{equation}\label{connectionform}
\A_\phi(X)= \sum_jP_j\phi^\dagger X P_j(\psi^\dagger\psi)^{-1}.
\end{equation}
To see this we first note that each $P_j\phi^\dagger X P_j(\psi^\dagger\psi)^{-1}$ belongs to $\u(\psi)$. For
\begin{equation}\label{eq}
	P_j\phi^\dagger X P_j(\psi^\dagger\psi)^{-1}=(\psi^\dagger\psi)^{-1}P_j\phi^\dagger X P_j
\end{equation}
implies that $P_j\phi^\dagger X P_j(\psi^\dagger\psi)^{-1}$ commutes with $\psi^\dagger\psi$, and $X^\dagger\phi+\phi^\dagger X=\0$ implies that $P_j\phi^\dagger X P_j(\psi^\dagger\psi)^{-1}$ 
is skew-Hermitian:
\begin{equation}
	\left(P_j\phi^\dagger X P_j(\psi^\dagger\psi)^{-1}\right)^\dagger+P_j\phi^\dagger X P_j(\psi^\dagger\psi)^{-1}=P_j(X^\dagger\phi+\phi^\dagger X)P_j(\psi^\dagger\psi)^{-1}=\0.
\end{equation}
The formula \eqref{connectionform} now follows from the calculation
\begin{equation}
	\begin{split}
		\II_{\phi}\Big(\sum_jP_j&\phi^\dagger X P_j(\psi^\dagger\psi)^{-1}\Big)\eta =\\
		&=\hbar\sum_j\Tr\big(P_jX^\dagger\phi P_j(\psi^\dagger\psi)^{-1}\phi^\dagger\phi\eta +\eta^\dagger\phi^\dagger\phi(\psi^\dagger\psi)^{-1} P_j \phi^\dagger X P_j\big)\\
		&=\hbar\sum_j\Tr\big(P_j(X^\dagger\phi\eta+\eta^\dagger\phi^\dagger X)P_j\big)\\
		&=\hbar\Tr\big(X^\dagger\phi\eta+\eta^\dagger\phi^\dagger X\big)\\
		&=\JJ_\phi(X)\eta,
	\end{split}
\end{equation}
where in the first identity we have used that $P_j\phi^\dagger X P_j$ commutes with $(\psi^\dagger\psi)^{-1}$, in the second identity that $\phi^\dagger\phi=\psi^\dagger\psi$, and in the third identity that $\eta P_j=P_j\eta$ for every $\eta$ in $\u(\psi)$.
As expected, the bundle map in Proposition \ref{relation} preserves the mechanical connection form: 
\begin{proposition}\label{preservation}
	$R_U^*\A_\phi=\Ad_{U^\dagger}\circ \A_\psi$.
\end{proposition}
\begin{proof}
If $P_j$ is the orthogonal projection onto the $j^\text{th}$ eigenspace of $\psi^\dagger\psi$,
then $U^\dagger P_j U$ is the orthogonal projection onto the $j^\text{th}$ eigenspace of $\phi^\dagger\phi$. 
Now
\begin{equation}
\begin{split}
	\A_{\phi}(dR_U(X))
	&=\sum_j(U^\dagger P_j U)\phi^\dagger (XU) (U^\dagger P_jU) (\phi^\dagger\phi)^{-1}\\
	&=U^\dagger\Big(\sum_jP_j \psi^\dagger X P_j(\psi^\dagger\psi)^{-1}\Big)U\\
	&=(\Ad_{U^\dagger}\circ \A_\psi)(X).
\end{split}
\end{equation}
\end{proof}
\noindent We use the following bijective correspondence between the eigenvectors of $\rho$ and $\psi^\dagger\psi$
to show that \eqref{geometriphase} generalizes definition \eqref{geophase}:
\begin{itemize}
\item If $\ket{\psi_k}$ is a normalized eigenvector of $\rho$ with eigenvalue $p_k$, then $\psi^\dagger\ket{\psi_k}/\sqrt{p_k}$ is a normalized eigenvector of $\psi^\dagger\psi$ with eigenvalue $p_k$.
\item If $\ket{k}$ is a normalized eigenvector of $\psi^\dagger\psi$ with eigenvalue $p_k$, then $\psi\ket{k}/\sqrt{p_k}$ is a normalized eigenvector of $\rho$ with eigenvalue $p_k$.
\end{itemize}
\begin{proposition}\label{horiz}
Let $\rho(t)$ be a curve in $\D(\rho)$ that emanates from $\rho$, and $\psi(t)$ be a lift of $\rho(t)$ to $\S(\psi)$ that emanates from $\psi$.
Then $\psi(t)$ is horizontal if and only if for every pair of eigenvectors $\ket{k}$ and $\ket{l}$ of $\psi^\dagger\psi$ with common eigenvalue, say $p$, the curves $\ket{\psi_k(t)}=\psi(t)\ket{k}/\sqrt{p}$ and $\ket{\psi_l(t)}=\psi(t)\ket{l}/\sqrt{p}$ satisfy
$\braket{\psi_k(t)}{\dot{\psi_l}(t)}=0$.
\end{proposition}
\begin{proof}
Let $P_j$ be the orthogonal projection onto the $j^\text{th}$ eigenspace of $\psi^\dagger\psi$.
Then 
\begin{equation}
	p\braket{\psi_k(t)}{\dot{\psi_l}(t)}
	=\bra{k}\psi(t)^\dagger\dot{\psi}(t)\ket{l}
	= \sum_j\bra{k}P_j\psi(t)^\dagger\dot{\psi}(t)P_j\ket{l}.
\end{equation}
The assertion now follows from \eqref{connectionform} and the aforementioned bijective correspondence between the eigenvectors of $\rho$ and $\psi^\dagger\psi$.
\end{proof}
From Proposition \ref{horiz} we conclude if $U(t)$ parallel transports $\rho$, and $\psi(t)=U(t)\psi$, then $\psi(t)$ is a horizontal lift of $\rho(t)=U(t)\rho U(t)^\dagger$, and 
\begin{equation}
	\Tr(U(\tau)\rho)=\Tr(U(\tau)\psi(0)\psi(0)^\dagger)=\Tr(\psi(0)^\dagger\psi(\tau)).
\end{equation}
Thus the geometric phases \eqref{geophase} and \eqref{geometriphase} agree.

As the evolution operator of a quantum system is not always parallel transporting, it is desirable to have an expression for the geometric phase which only involves a purification of the initial state and the, possibly non-parallel transporting, evolution operator. The explicit formula \eqref{connectionform} makes it possible to derive such an expression. For if $\psi(t)$ is a curve in $\S(\psi)$, and
\begin{equation}
	\phi(t)=\psi(t)\ptexp\left(-\int_0^t\ds\,\A_{\psi(s)}(\dot{\psi}(s))\right),
\end{equation}
where $\ptexp$ is the positively time-ordered exponential, then $\phi(t)$ is a horizontal curve in $\S(\psi)$ 
that projects onto the same curve as $\psi(t)$.
Thus if $\rho(t)=U(t)\rho U(t)^\dagger$,
\begin{equation}
	\gamma_g[\rho(t)]
	=\arg\Tr\Bigg(\psi^\dagger U(\tau)\psi\ptexp\Big(-\int_0^\tau\ds\sum_jP_j\psi^\dagger U(s)^\dagger\dot{U}(s)\psi P_j(\psi^\dagger\psi)^{-1}\Big)\Bigg).
\end{equation}
\begin{example}[Geometric phase of a non-parallel transported mixed state]\label{easy}
Consider an ensemble of qubits represented by the density matrix $\rho=\diag(p_1,p_2)$, $p_1\ne p_2$, and assume that the qubits are affected by the Hamiltonian $\hat{H}=-\omega\mathbf{n}\cdot\boldsymbol\sigma$, where $\mathbf{n}=(\sin\theta,0,\cos\theta)^T$ and $\boldsymbol\sigma$ is the vector of Pauli matrices.
The evolution operator associated with $\hat{H}$ is 
\begin{equation*}
	U(t)=\cos(\omega t)\begin{pmatrix} 
	1 & 0\\ 
	0 & 1
	\end{pmatrix}+i\sin(\omega t)
	\begin{pmatrix} 
	\cos\theta & \sin\theta\\ 
	\sin\theta & -\cos\theta
	\end{pmatrix},
\end{equation*}
which is non-parallel transporting for $\rho$ if $\cos\theta\ne 0$, and a
lift of the ensemble's evolution curve is given by $\psi(t)=U(t)\psi$, where
$\psi=\diag(\sqrt{p_1},\sqrt{p_2})$. For this lift, 
\begin{equation}
	\ptexp\left(-\int_0^t\ds\,\A_{\psi(s)}\left(\dot\psi(s)\right)\right)
	=\begin{pmatrix} 
	e^{i\omega t\cos\theta} & 0\\ 
	0 & e^{-i\omega t\cos\theta}
	\end{pmatrix}.
\end{equation}
Thus the geometric phase of the evolution curve is
\begin{equation}
\begin{split}
	\gamma_g[\rho(t)]
	&=\arg\Tr\left(\psi^\dagger U(\tau)\psi\ptexp\left(-\int_0^t\ds\,\A_{\psi(s)}(\dot\psi(s))\right)\right)\\
	&=\arg\Tr \begin{pmatrix} 
	          p_1(\cos(\omega\tau)+i\sin(\omega\tau)\cos\theta)e^{i\omega\tau\cos\theta}\qquad\qquad *\qquad\qquad\\ 
	          \qquad\qquad*\qquad\qquad p_2(\cos(\omega\tau)-i\sin(\omega\tau)\cos\theta)e^{-i\omega\tau\cos\theta}
	          \end{pmatrix}.
\end{split}
\end{equation}
If, in particular, $\tau$ equals $\pi/\omega$,
the curve $\rho(t)$ is a loop, and
\begin{equation}
	\gamma_g[\rho(t)]
	=\arg\left(p_1 e^{i\pi(1+\cos\theta)}+p_2 e^{i\pi(1-\cos\theta)}\right).
\end{equation}
\end{example}
	
\section{Geometric phase for arbitrarily evolving quantum systems}\label{section:open}
In the previous section we saw that the standard purification bundle can be reduced to a bundle with a naturally defined connection giving rise to the geometric phase of Sj\"oqvist \emph{et al.}
In \cite{Tong_etal2005}, Tong \emph{et al.}~generalized the geometric phase of Sj\"oqvist \emph{et al.}~to systems for which the evolution is not necessarily governed by unitary operators.
We will in this section show that the standard purification bundle can be equipped with a natural connection which gives rise to the geometric phase of Tong \emph{et al.}
	
\subsection{Stratification of the standard purification bundle}
Let $\psi$ be an invertible operator on $\HH$.
The right coadjoint orbit of $\J(\psi)$ in $\u(\HH)^*$ is
\begin{equation}
	\O_{\J(\psi)}=\{\Ad^*_U\J(\psi): U\in\U(\HH)\}.
\end{equation}
This orbit is a submanifold of $\u(\HH)^*$, and $\J$ is transversal to it by Proposition \ref{smooth}. 
Thus the preimage $\J^{-1}(\O_{\J(\psi)})$ is a submanifold of $\L^\circ(\HH)$, which we for simplicity will denote by $\mathcal{Q}(\psi)$. 
If $\psi$ is a purification, then $\Q(\psi)$ is contained in $\S^\circ(\HH)$, and if $\psi$ purifies $\rho$, then $\Q(\psi)=\Pi^{-1}(\D(\rho))$.
In fact, $\Q(\psi)$ is foliated by those $\S(\phi)$ for which $\phi=U\psi V$  for some $U$ and $V$ in $\U(\HH)$, see Figure \ref{figure:stratification}.
\begin{figure}
	\centering
	\includegraphics[width=0.75\textwidth]{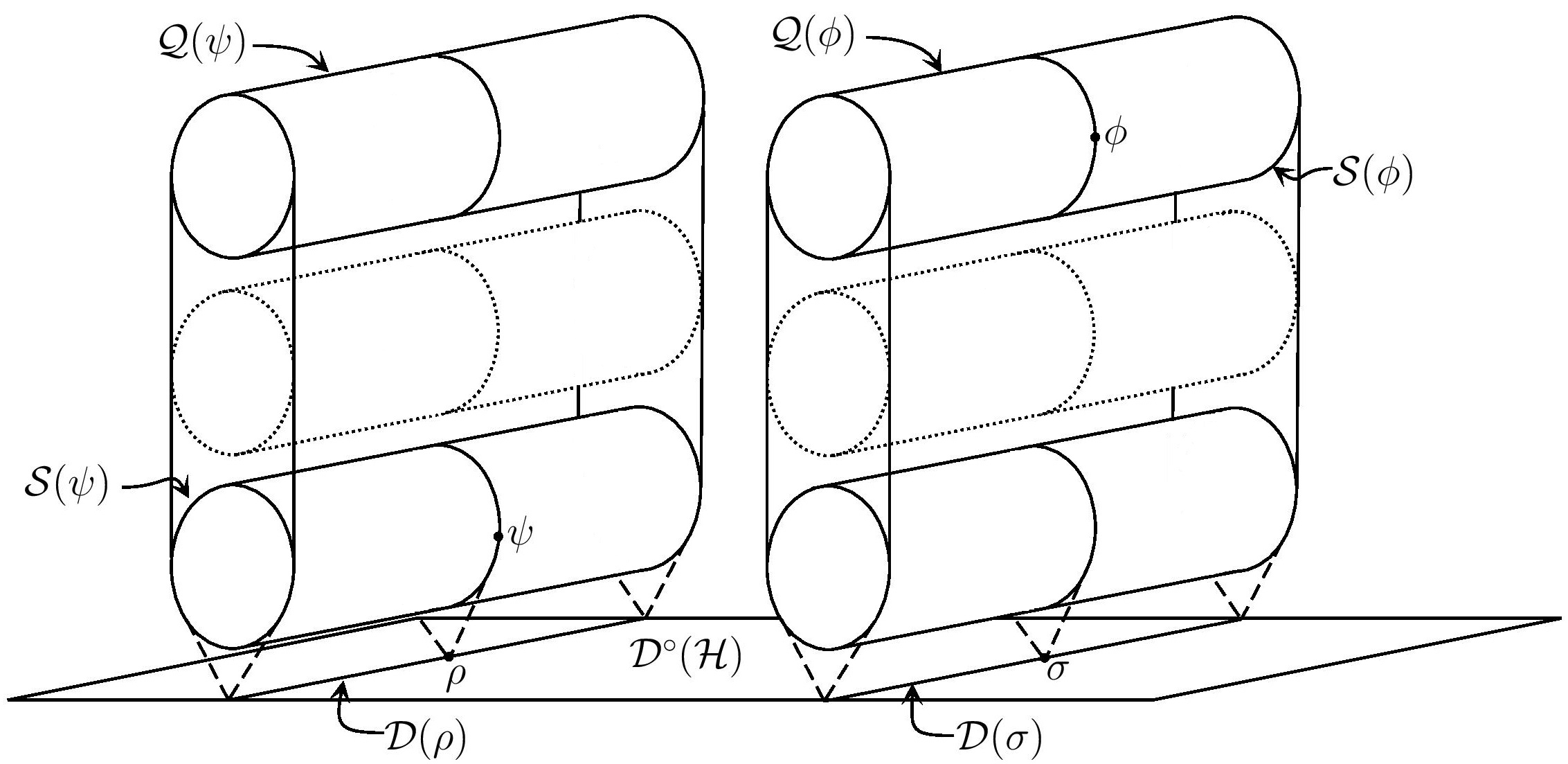}
	\caption{The stratification of $\S^\circ(\HH)$.}
	\label{figure:stratification}
\end{figure}
Next we show that the tangent space of $\Q(\psi)$ at $\psi$ is the sum of Uhlmann's vertical 
space at $\psi$ and the kernel of the differential of $\J$.

Recall that 
\begin{equation}
	\V_\psi\S^\circ(\HH)=\{X\in \T_\psi\S^\circ(\HH):  X\psi^\dagger+\psi X^\dagger =0\}=\{\hat{\xi}(\psi):\xi\in\u(\HH)\}.
\end{equation}
Moreover, for each $X$ in $\T_\psi\L^\circ(\HH)$ and $\xi$ in $\u(\HH)$,
\begin{equation}
	d\J(X)\xi=\Omega(\hat{\xi}(\psi),X)=i\hbar\Tr((X^\dagger\psi+\psi^\dagger X)\xi).
\end{equation}
Therefore, $X$ is in the kernel of the differential of $\J$ at $\psi$ if and only if $X^\dagger\psi+\psi^\dagger X =0$:
\begin{equation}
	\Ker d\J_\psi=\{X\in \T_\psi\L^\circ(\HH): X^\dagger\psi+\psi^\dagger X =0\}\label{kernel}.
\end{equation}
	
\begin{proposition}\label{dimensionofQ}
	$\T_\psi\Q(\psi)=\V_\psi\S^\circ(\HH)+\Ker d\J_\psi$.
\end{proposition}
\begin{proof}
Let $n$ be the \emph{complex} dimension of $\HH$, $n_1,n_2,\dots,n_l$ be the multiplicities of the different eigenvalues of $\psi^\dagger\psi$, and set $q=n_1^2+n_2^2+\dots + n_l^2$. According to Proposition \ref{smooth}, the codimension of $\Q(\psi)$ in $\L^\circ(\HH)$ equals the codimension of $\O_{\J(\psi)}$ in $\u(\HH)^*$, namely $q$. Therefore, the dimension of $\Q(\psi)$ is $2n^2-q$.
	
The dimension of $\V_\psi\S^\circ(\HH)$ equals the dimension of the symmetry group $\U(\HH)$, which is $n^2$,
and the dimension of $\Ker d\J_\psi$ equals the difference of the dimensions of $\L^\circ(\HH)$ and $\u(\HH)$, which is also $n^2$.
Now, the assertion follows from the observation that the intersection of $\V_\psi\S^\circ(\HH)$ and $\Ker d\J_\psi$ equals $\V_\psi\S(\psi)$, which has dimension $q$. Indeed, by \eqref{kernel}, for each $\xi$ in $\u(\HH)$, the vector $\hat{\xi}(\psi)$ is in the kernel of the differential of $\J$ if and only if $\xi$ commutes with $\psi^\dagger\psi$.
\end{proof}
\subsection{A connection in the standard purification bundle}
The decomposition of $\S^\circ(\HH)$ into submanifolds $\Q(\psi)$, and the decomposition of each $\Q(\psi)$ into submanifolds $\S(\phi)$, suggests a new connection in the standard purification bundle: For each purification $\psi$ let $\K_\psi\S^\circ(\HH)$ be the space of all vectors in $\T_\psi\S^\circ(\HH)$ which are perpendicular to $\T_\psi\Q(\psi)$, and set 
\begin{equation}
\LL_\psi\S^\circ(\HH)= \H_\psi\S(\psi)\oplus \K_\psi\S^\circ(\HH),
\end{equation}
see Figure \ref{figure:newconnection}.
\begin{figure}
	\centering
	\includegraphics[width=0.65\textwidth]{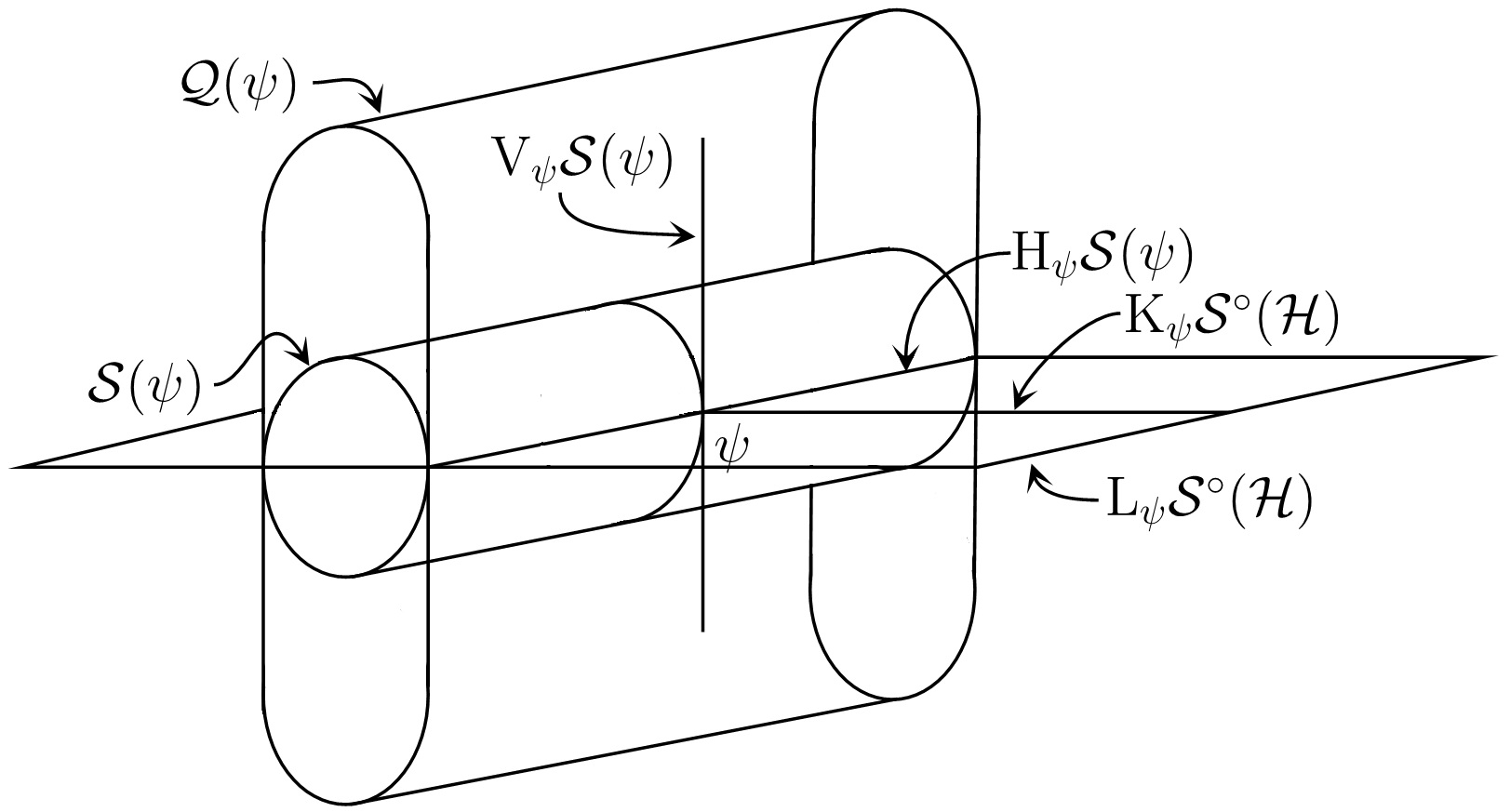}
	\caption{The horizontal space at $\psi$.}
	\label{figure:newconnection}
\end{figure}
Then $\LL\S^\circ(\HH)=\bigcup_{\psi\in\S^\circ(\HH)} \LL_\psi\S^\circ(\HH)$ is a connection in the standard purification bundle.
Before we prove this we observe that
\begin{itemize}
	\item $\K_\psi\S^\circ(\HH)$ is a subspace of $\H_\psi\S^\circ(\HH)$ because $\V_\psi\S^\circ(\HH)$ is a subspace of $\T_\psi\Q(\psi)$.
	\item The sum of $\H_\psi\S(\psi)$ and $\K_\psi\S^\circ(\HH)$ is indeed direct by Proposition \ref{infskillnad}.
	\item The spaces $\H_\psi\S(\psi)$ and $\V_\psi\S^\circ(\HH)$ only have the zero-vector in common because $\V_\psi\S(\psi)$ equals the intersection of $\T_\psi\S(\psi)$ and $\V_\psi\S^\circ(\HH)$.
\end{itemize}
\begin{proposition}\label{dimension}
	$\T_\psi\S^\circ(\HH)=\V_\psi\S^\circ(\HH)\oplus \LL_\psi\S^\circ(\HH)$.
\end{proposition}
\begin{proof}
Let $n$ and $q$ be as in the proof of Proposition \ref{dimensionofQ}.
Then the dimension of $\T_\psi\S^\circ(\HH)$ is $2n^2-1$, the dimension of $\V_\psi\S^\circ(\HH)$ is $n^2$, and 
the dimension of $\H_\psi\S(\psi)$ is the same as the dimension of $\D(\psi\psi^\dagger)$, namely $n^2-q$.
Moreover, the dimension of $\K_\psi\S^\circ(\HH)$ equals the codimension of $\Q(\psi)$ in $\S^\circ(\HH)$,  which is $q-1$. Thus
\begin{equation}
\begin{split}
	\dim\V_\psi\S^\circ(\HH)\oplus \H_\psi\S(\psi)\oplus \K_\psi\S^\circ(\HH)
	&=n^2+(n^2-q)+(q - 1)\\
	&=2n^2-1\\
	&=\dim\T_\psi\S^\circ(\HH).
\end{split}
\end{equation}
\end{proof}
\begin{proposition}\label{gaugeinvariant}
$\LL\S^\circ(\HH)$ is invariant under the right action by $\U(\HH)$.
\end{proposition}
\begin{proof}
Let $U$ be a unitary operator on $\HH$. Then 
$dR_U(\H_\psi\S(\psi))=\H_{\psi U}\S(\psi U)$
according to Proposition \ref{preservation}.
Moreover, 
$dR_U(\K_\psi\S^\circ(\HH))=\K_{\psi U}\S^\circ(\HH)$ 
because $R_U$ is an isometry on $\L^\circ(\HH)$ which preserves both $\Q(\psi)$ and $\S^\circ(\HH)$.
\end{proof}
\subsection{Geometric phase for quantum ensembles}
Let $\rho(t)$ be a curve in $\D^\circ(\HH)$.
We define the geometric phase of $\rho(t)$ to be 
\begin{equation}\label{gp}
	\gamma_g[\rho(t)]=\arg\Tr\left(\psi(0)^\dagger\psi(\tau)\right),
\end{equation}
where $\psi(t)$ is any horizontal lift of $\rho(t)$ to $\S^\circ(\HH)$.
Here, and henceforth, by horizontal we mean that $\dot{\psi}(t)$ belongs to $\LL_{\psi(t)}\S^\circ(\HH)$ for every $0\leq t\leq \tau$. 
	
The definition \eqref{gp} is manifestly invariant under the action of the symmetry group. Moreover, it reduces to definition \eqref{geometriphase} if $\rho(t)$ is contained in the orbit $\D(\rho)$, where $\rho=\rho(0)$.
For then $\psi(t)$ is contained in $\S(\psi)$, where $\psi=\psi(0)$, and $\dot{\psi}(t)$
sits in $\H_{\psi(t)}\S(\psi)$. Next we show that for an arbitrary $\rho(t)$, the definitions \eqref{gp} and \eqref{opengp} agree. 
To this end, let 
\begin{equation}\label{r}
	\rho(t)=\sum_k p_k(t)\ketbra{\psi_k(t)}{\psi_k(t)}
\end{equation} 
be a smoothly varying spectral decomposition of $\rho(t)$. Furthermore, let $\{\ket{k}\}$ be an orthonormal basis for $\HH$, and define 
\begin{equation}\label{p}
	\psi(t)=\sum_k\sqrt{p_k(t)}\ketbra{\psi_k(t)}{k}.
\end{equation} 
Then $\psi(t)$ is a lift of $\rho(t)$ to $\S^\circ(\HH)$, see Figure \ref{figure:velocitydecomp}.
\begin{figure}
	\centering
	\includegraphics[width=0.75\textwidth]{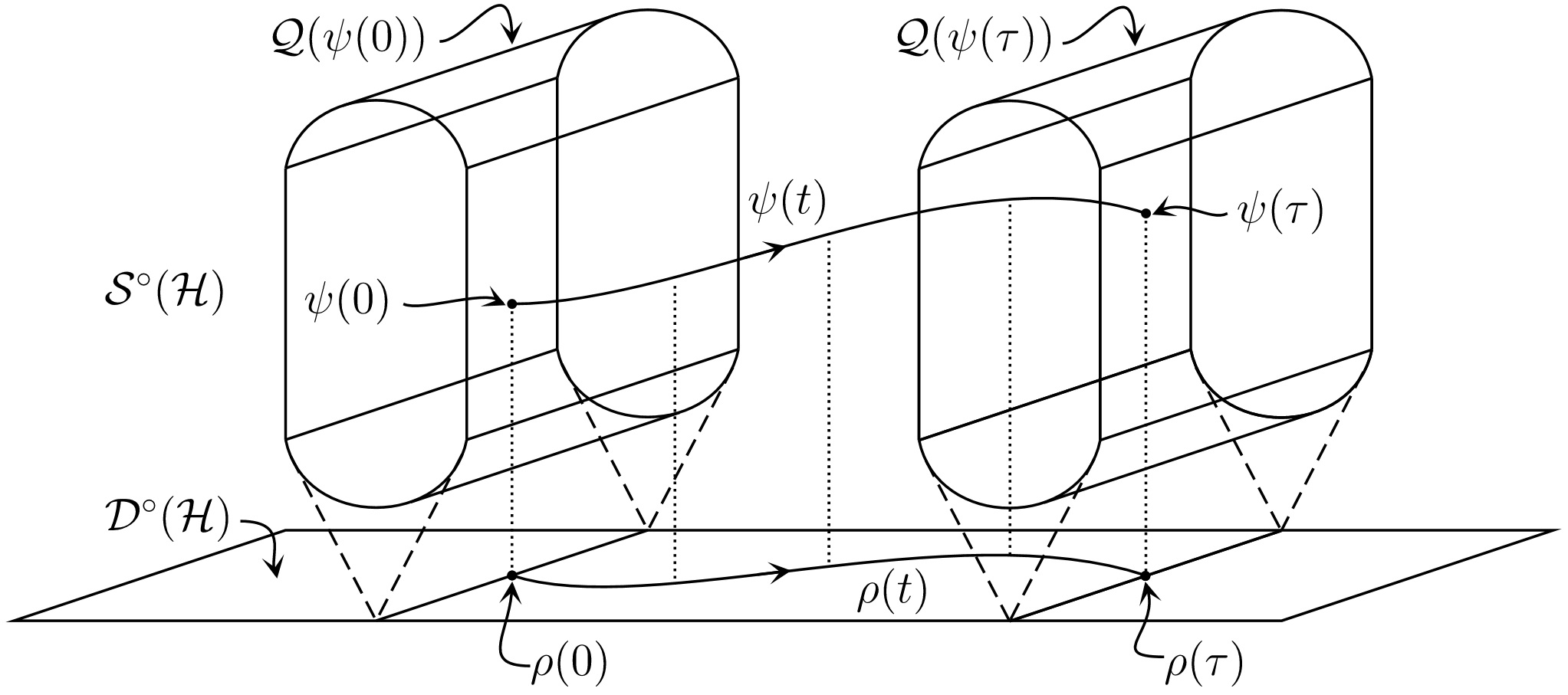}
	\caption{An evolution connecting non-isospectral mixed states.}
	\label{figure:velocitydecomp}
\end{figure}
\begin{theorem}\label{main}
The lift $\psi(t)$ is horizontal if and only if the curves $\ket{\psi_k(t)}$ satisfy the parallelism condition \eqref{pc}.
\end{theorem}
	
Define a $2$-parameter family of purifications by
\begin{equation}
	\Psi(u,v)=\sum_k\sqrt{p_k(v)}\ketbra{\psi_k(u)}{k}.
\end{equation}
Then $\Psi(t,t)=\psi(t)$, and the velocity field of $\psi(t)$ is 
\begin{equation}
	\dot{\psi}(t)=\pa_u\Psi(t,t)+\pa_v\Psi(t,t),
\end{equation} 
where 
\begin{align}
	\pa_u\Psi(t,t)&=\sum_k\sqrt{p_k(t)}\ketbra{\dot{\psi_k}(t)}{k},\label{u}\\
	\pa_v\Psi(t,t)&=\sum_k\dd{t}\left(\!\sqrt{p_k(t)}\right)\ketbra{\psi_k(t)}{k}.\label{v} 
\end{align}
We show that $\pa_u\Psi(t,t)$ belongs to $\T_{\psi(t)}\S(\psi(t))$ and $\pa_v\Psi(t,t)$ belongs to $\K_{\psi(t)}\S^\circ(\HH)$.
Theorem \ref{main} then follows from the observation 
\begin{equation}
\sqrt{p_k(t)p_l(t)}\braket{\psi_k(t)}{\dot{\psi}_l(t)}=\bra{k}\psi(t)^\dagger\pa_u\Psi(t,t)\ket{l}.
\end{equation}
For if $p_k(t)=p_l(t)$, then 
\begin{equation}
\begin{split}
\bra{k}\psi(t)^\dagger\pa_u\Psi(t,t)\ket{l}
&=p_l(t)\bra{k}\sum_jP_j(t)\psi(t)^\dagger\pa_u\Psi(t,t)P_j(t)(\psi(t)^\dagger\psi(t))^{-1}\ket{l}\\
&=p_l(t)\bra{k}\A_{\psi(t)}(\pa_u\Psi(t,t))\ket{l}.
\end{split}
\end{equation}
Here, $P_j(t)$ is orthogonal projection onto the $j^\text{th}$ eigenspace of $\psi(t)^\dagger\psi(t)$.
\begin{proposition}\label{T}
$\pa_u\Psi(t,t)$ belongs to $\T_{\psi(t)}\S(\psi(t))$.
\end{proposition}
\begin{proof}
The assertion follows from the observation that $u\mapsto \Psi(u,t)$ is a 
curve in $\S(\psi(t))$ which passes through $\psi(t)$ when $u=t$.
\end{proof}
\begin{proposition}\label{K}
$\pa_v\Psi(t,t)$ belongs to $\K_{\psi(t)}\S^\circ(\HH)$. 
\end{proposition}
\begin{proof}
First we verify that $\pa_v\Psi(t,t)$ belongs to $\H_{\psi(t)}\S^\circ(\HH)$, and hence is perpendicular to $\V_{\psi(t)}\S^\circ(\HH)$:
\begin{equation}
\begin{split}
	\pa_v&\Psi(t,t)^\dagger\psi(t)-\psi(t)^\dagger \pa_v\Psi(t,t)=\\
	&=\sum_{k,l}\left(\sqrt{p_l(t)}\dd{t}\left(\!\sqrt{p_k(t)}\right)-\sqrt{p_k(t)}\dd{t}\left(\!\sqrt{p_l(t)}\right)\right)\ket{k}\braket{\psi_k(t)}{\psi_l(t)}\bra{l}=0.
\end{split}
\end{equation}
Next we prove that $\pa_v\Psi(t,t)$ is also perpendicular to $\Ker d\J_{\psi(t)}$.
For this, let $\phi(s)$ be a curve in $\L^\circ(\HH)$ such that $\phi(t)=\psi(t)$ and $d\J(\dot{\phi}(t))=0$. Write 
\begin{equation}
	\phi(s)=\sum_k\ketbra{\phi_k(s)}{k}.
\end{equation} 
According to \eqref{kernel},  
\begin{equation}
    \begin{split}
        0
	    &=\dot{\phi}(t)^\dagger\psi(t)+\psi(t)^\dagger\dot{\phi}(t)\\
	    &=\sum_{k,l}\ket{k}\Big(\sqrt{p_l(t)}\braket{\dot{\phi_k}(t)}{\psi_l(t)}+\sqrt{p_k(t)}\braket{\psi_k(t)}{\dot{\phi_l}(t)}\Big)\bra{l}.
	\end{split}
\end{equation}
Thus $\braket{\dot{\phi_k}(t)}{\psi_k(t)}+\braket{\psi_k(t)}{\dot{\phi_k}(t)}=0$ for every $k$. Now,
\begin{equation}
	\begin{split}
	\pa_v&\Psi(t,t)^\dagger\dot{\phi}(t)+\dot{\phi}(t)^\dagger \pa_v\Psi(t,t)=\\
	&=\sum_{k,l}\ket{k}\Bigg(\dd{t}\left(\!\sqrt{p_k(t)}\right)\braket{\psi_k(t)}{\dot{\phi_l}(t)}+\dd{t}\left(\!\sqrt{p_l(t)}\right)\braket{\dot{\phi_k}(t)}{\psi_l(t)}\Bigg)\bra{l}.
	\end{split}
\end{equation}
Consequently,
\begin{equation}
	G(\pa_v\Psi(t,t),\dot{\phi}(t))=\hbar\sum_k\dd{t}\left(\!\sqrt{p_k(t)}\right)\Big(\braket{\psi_k(t)}{\dot{\phi_k}(t)}+\braket{\dot{\phi_k}(t)}{\psi_k(t)}\Big)=0.
\end{equation}
The assertion follows from Proposition \ref{dimensionofQ}.
\end{proof}

It is now straightforward to show that the two definitions \eqref{opengp} and \eqref{gp} agree.
For if $\rho(t)$ is given by \eqref{r}, and the curves $\ket{\psi_k(t)}$ satisfy the parallelism condition \eqref{pc}, we define $\psi(t)$ as in \eqref{p}. According to Theorem \ref{main}, then,  $\psi(t)$ is a horizontal lift of $\rho(t)$, and by definition \eqref{gp},
\begin{equation}
	\begin{split}
		\gamma_g[\rho(t)]
		&=\arg\Tr\left(\psi(0)^\dagger\psi(\tau)\right)\\
		&=\arg\Tr\sum_{k,l}\sqrt{p_k(0)p_l(\tau)}\ket{k}\braket{\psi_k(0)}{\psi_l(\tau)}\bra{l}\\
		&=\arg\sum_k\sqrt{p_k(0)p_k(\tau)}\braket{\psi_k(0)}{\psi_k(\tau)}.
	\end{split}
\end{equation}

\begin{example}
Consider a qubit system which with respect to a computational basis $\{\ket{1},\ket{2}\}$ is described by the curve of density operators
\begin{equation}
\rho(t)=\frac{1}{2}
\begin{pmatrix}
1+z(t) & x(t)-iy(t) \\ x(t)+iy(t) & 1-z(t)
\end{pmatrix},\quad (0\leq t\leq 2\pi),
\end{equation}
where 
\begin{equation}
x(t)=\frac{1}{2}\cos t,\quad 
y(t)=\frac{1}{4}\sin t,\quad
z(t)=\frac{1}{2}\sin^2(t/2).
\end{equation}
Figure \ref{figure:trefoil} shows the curve in the Bloch ball traced out by the Bloch vector of  $\rho(t)$.
\begin{figure}
	\centering
	\includegraphics[width=0.35\textwidth]{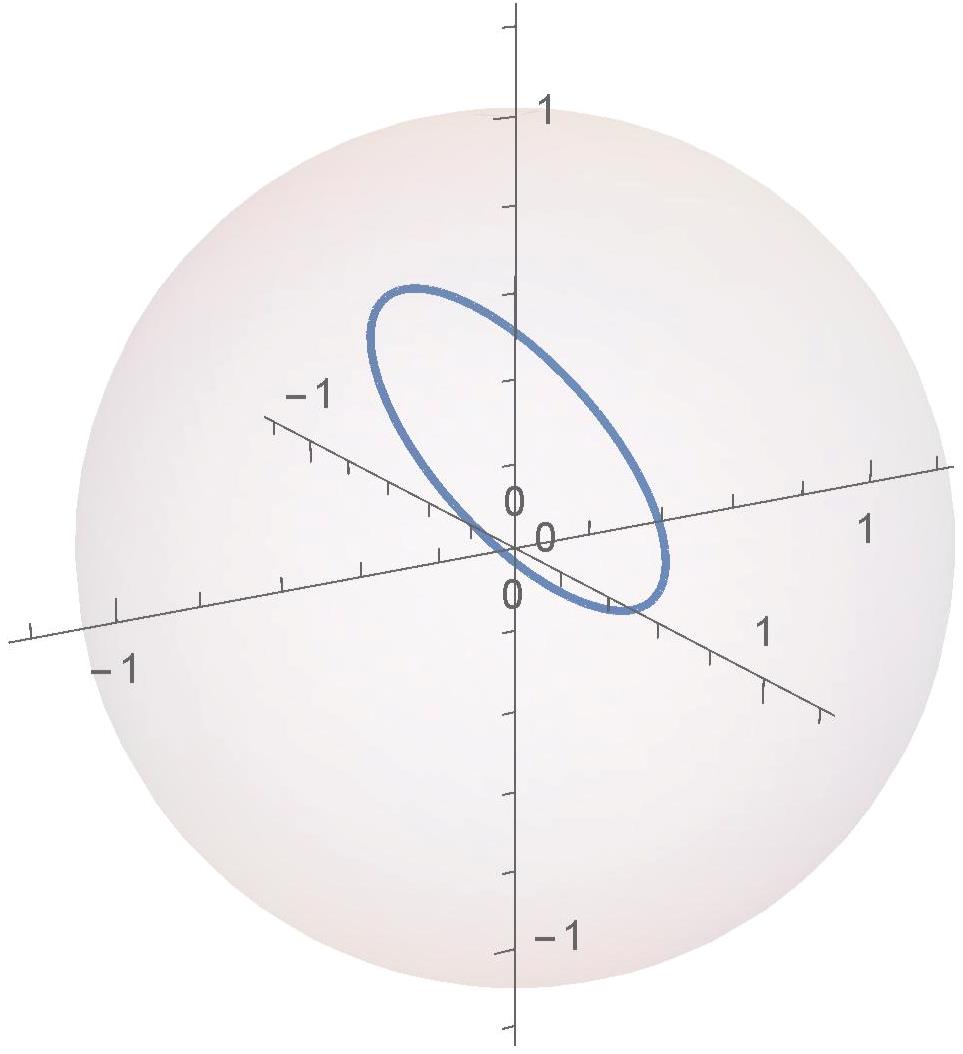}
	\caption{A curve in the Bloch ball.}
	\label{figure:trefoil}
\end{figure}
This vector has at each instant $t$ the length
\begin{equation}
r(t)=\sqrt{x(t)^2+y(t)^2+z(t)^2}=\frac{1}{4}\sqrt{4\cos^2 t+\sin^2 t+4\sin^4(t/2)},
\end{equation}
and standard formulas tells us that a spectral decomposition for $\rho(t)$ is given by 
\begin{equation}
\rho(t)=p_1(t)\ketbra{\psi_1(t)}{\psi_1(t)}+p_2(t)\ketbra{\psi_2(t)}{\psi_2(t)},
\end{equation}
where 
\begin{equation}
p_1(t)=\frac{1}{2}(1+r(t)),\quad p_2(t)=\frac{1}{2}(1-r(t)),
\end{equation}
and 
\begin{align}
& \ket{\psi_1(t)}=\frac{1}{\sqrt{2r(t)(r(t)+z(t))}}
\begin{pmatrix}
r(t)+z(t)\\
x(t)+iy(t)
\end{pmatrix},\\
& \ket{\psi_2(t)}=\frac{1}{\sqrt{2r(t)(r(t)+z(t))}}
\begin{pmatrix}
-x(t)+iy(t)\\
r(t)+z(t)
\end{pmatrix}.
\end{align}
A lift of $\rho(t)$ is then given by
\begin{equation}\label{tre}
\psi(t)=\sqrt{p_1(t)}\ketbra{\psi_1(t)}{1}+\sqrt{p_2(t)}\ketbra{\psi_2(t)}{2}.
\end{equation}
This lift, however, is not horizontal. But to make it horizontal 
it is sufficient, according to Theorem \ref{main}, to replace 
the eigenvectors in \eqref{tre} by 
a pair of phase shifted ones,
\begin{equation}
\ket{\phi_1(t)}=\ket{\psi_1(t)}e^{i\theta_1(t)},\quad 
\ket{\phi_2(t)}=\ket{\psi_2(t)}e^{i\theta_2(t)},
\end{equation}
where the phases are such that the new eigenvectors fulfill the parallel transport condition
$\braket{\phi_1(t)}{\dot{\phi}_1(t)}=\braket{\phi_2(t)}{\dot{\phi}_2(t)}=0$.
The so obtained curve
\begin{equation}
\phi(t)=\sqrt{p_1(t)}\ketbra{\phi_1(t)}{1}+\sqrt{p_2(t)}\ketbra{\phi_2(t)}{2}
\end{equation}
\emph{is} a horizontal lift of $\rho(t)$, and by \eqref{gp}, the geometric phase of $\rho(t)$ equals
\begin{equation}
\gamma_g[\rho(t)]=
\arg\left(\sqrt{p_1(0)p_1(2\pi)}\braket{\phi_1(0)}{\phi_1(2\pi)}
+\sqrt{p_2(0)p_2(2\pi)}\braket{\phi_2(0)}{\phi_2(2\pi)}\right).
\end{equation}
Suitable phase factors are
\begin{equation}
\theta_1(t)=i\int_{0}^{t}\ds\,\braket{\psi_1(s)}{\dot{\psi}_1(s)},\quad \theta_2(t)=i\int_{0}^{t}\ds\,\braket{\psi_2(s)}{\dot{\psi}_2(s)},
\end{equation}
and a numerical calculation yields
\begin{equation}
\gamma_g[\rho(t)]\approx
\arg\left(0.17 - 0.49 i\right)\approx 5.0.
\end{equation}
\end{example}

\section{Summary and outlook}
In this paper we have examined the geometry of Uhlmann's standard purification bundle using tools from the theory of symmetries of dynamical systems. The examination revealed the existence of a natural connection in the standard purification bundle which is different from that of Uhlmann, and which gives rise to the geometric phase for mixed quantum states proposed by Sj\"{o}qvist \emph{et al.}~and Tong \emph{et al.}
The approach used here is purely abstract, and thus it shows that the definition of the geometric phase of Sjöqvist-Tong \emph{et al.}, fundamentally, does not rest upon the ability to describe evolving mixed states as curves of incoherent superpositions of pure states.

Another approach to geometric phase which, at least at first glance, do rest upon the fact that mixed states can be written as \emph{homogeneous} incoherent superpositions of pure states was proposed in \cite{Rezakhani_etal2006}. In that paper a unifying formalism for geometric phase was developed to illuminate the relationship between Uhlmann's phase and the phase of Sj\"oqvist-Tong \emph{et al.} The authors of the current paper would like to know, and intend to investigate whether the definition of geometric phase in \cite{Rezakhani_etal2006} can be derived from a gauge theory for mixed states which unifies Uhlmann's approach to geometric phase and the one presented here.

\vspace{10pt}
\end{document}